\documentclass[twoside]{article}
\usepackage{hyperref}
\usepackage{url}
\usepackage{amsthm}
\usepackage{amsmath}
\usepackage{amsfonts}
\usepackage{mathrsfs}
\usepackage{algorithm}
\usepackage[noend]{algorithmic}
\usepackage{bm}
\usepackage{pstricks}
\usepackage{tikz}
\usepackage{subcaption}
\usepackage{wrapfig}
\usepackage{mathtools}

\usepackage[accepted]{icml2016-mod}

\newcommand{\lb}{\operatorname{lb}}
\newcommand{\ub}{\operatorname{ub}}
\newcommand{\pruned}{\operatorname{pruned}}
\newcommand{\canchange}{\operatorname{canchange}}
\newcommand{\closest}{\operatorname{closest}}
\newcommand{\parent}{\operatorname{parent}}

\newcommand{\N}{\mathscr{N}}

\newtheorem{defn}{Definition}
\newtheorem{thm}{Theorem}
\newtheorem{lemma}{Lemma}

\begin{document}




\twocolumn[
\vspace*{-0.1in}
\icmltitle{Dual-tree $k$-means with bounded iteration runtime}

\vspace*{-0.5em}
\icmlauthor{Ryan R. Curtin}{ryan@ratml.org}
\icmladdress{School of Computational Science and Engineering, \\
Georgia Institute of Technology, Atlanta, GA 30332 USA}

\icmlkeywords{clustering, k-means, dual-tree algorithms, kd-tree, cover tree}

\vskip 0.15in
]

\begin{abstract}
$k$-means is a widely used clustering algorithm, but for $k$ clusters and a
dataset size of $N$, each iteration of Lloyd's algorithm costs $O(kN)$ time.
Although there are existing techniques to accelerate single Lloyd iterations,
none of these are tailored to the case of large $k$, which is increasingly
common as dataset sizes grow.  We propose a dual-tree algorithm that gives the
{\it exact} same results as standard $k$-means; when using cover trees, we use
adaptive analysis techniques to, under some assumptions, bound the
single-iteration runtime of the algorithm as $O(N + k \log k)$.  To our
knowledge these are the first sub-$O(kN)$ bounds for exact Lloyd iterations.  We
then show that this theoretically favorable algorithm performs competitively in
practice, especially for large $N$ and $k$ in low dimensions.  Further, the
algorithm is tree-independent, so any type of tree may be used.
\end{abstract}

\vspace*{-1.8em}
\section{Introduction}
\vspace*{-0.4em}

Of all the clustering algorithms in use today, among the simplest and most
utilized is the venerated $k$-means clustering algorithm, usually implemented
via Lloyd's algorithm: given a dataset $S$, repeat the following two steps (a
`Lloyd iteration') until the centroids of each of the $k$ clusters converge:

\vspace*{-1.0em}
\begin{enumerate} \itemsep -2pt
  \item Assign each point $p_i \in S$ to the cluster with nearest centroid.
  \item Recalculate the centroids for each cluster using the assignments of each
point in $S$.
\end{enumerate}
\vspace*{-1.0em}

Clearly, a simple implementation of this algorithm will take $O(kN)$ time where
$N = |S|$.  However, the number of iterations is not bounded unless the
practitioner manually sets a maximum, and $k$-means is not guaranteed to
converge to the global best clustering.  Despite these shortcomings, in practice
$k$-means tends to quickly converge to reasonable solutions.  Even so, there is
no shortage of techniques for improving the clusters $k$-means converges to:
refinement of initial centroids \cite{bradley1998refining} and weighted
sampling of initial centroids \cite{arthur2007k} are just two of many popular
existing strategies.

There are also a number of methods for accelerating the runtime of a single
iteration of $k$-means.  In general, these ideas use the triangle inequality to
prune work during the assignments step.  Algorithms of this sort include the
work of Pelleg and Moore \yrcite{pelleg1999accelerating}, Elkan
\yrcite{elkan2003using}, Hamerly \yrcite{hamerly2010making}, and Ding et~al.
\yrcite{ding2015yinyang}.  However, the scaling of these algorithms can make
them problematic for the case of large $k$ and large $N$.


\setlength{\textfloatsep}{0.4em}
\begin{table*}[t!]
\begin{center}
\begin{tabular}{|c|c|c|c|}
\hline
{\bf Algorithm} & {\bf Setup} & {\bf Worst-case} & {\bf Memory} \\
\hline
naive & n/a & $O(kN)$ & $O(k + N)$ \\
blacklist & $O(N \log N)$ & $O(kN)$ & $O(k \log N + N)$ \\
elkan & n/a & $O(k^2 + kN)$ & $O(k^2 + kN)$ \\
hamerly & n/a & $O(k^2 + kN)$ & $O(k + N)$ \\
yinyang & $O(k^2 + kN)$ & $O(kN)$ & $O(kN)$ \\
{\bf dualtree} & $O(N \log N)$ & $O(k \log k + N)^1$ & $O(k + N)$ \\
\hline
\end{tabular}
\end{center}
\vspace*{-1.0em}
\caption{Runtime and memory bounds for $k$-means algorithms.}
\label{tab:runtimes}
\vspace*{-1.0em}
\end{table*}

In this paper, we describe a dual-tree $k$-means algorithm tailored to the large
$k$ and large $N$ case that outperforms all competing algorithms in that
setting; this dual-tree algorithm also has bounded single-iteration runtime in
some situations (see Section \ref{sec:theory}).  This algorithm, which is our
main contribution, has several appealing aspects:

\vspace*{-0.5em}
\begin{itemize} \itemsep -1pt
  \item {\bf Empirical efficiency}.  In the large $k$ and large $N$ setting for
which this algorithm is designed, it outperforms all other alternatives, and
scales better to larger datasets.  The algorithm is especially efficient in
low dimensionality.

  \item {\bf Runtime guarantees}.  Using adaptive runtime analysis
techniques, we bound the single-iteration runtime of our algorithm with respect
to the intrinsic dimensionality of the centroids and data, when cover trees are
used.  This gives theoretical support for the use of our algorithm in large data
settings.  In addition, the bound is dependent on the intrinsic dimensionality,
{\it not} the extrinsic dimensionality.

  \item {\bf Generalizability}.  We develop our algorithm using a
tree-independent dual-tree algorithm abstraction \cite{curtin2013tree}; this
means that our algorithm may be used with {\it any} type of valid tree.  This
includes not just $kd$-trees but also metric trees, cone trees,
octrees, and others.  Different trees may be suited to different types of data,
and since our algorithm is general, one may use any type of tree as a
plug-and-play parameter.

  \item {\bf Separation of concerns}.  The abstraction we use to develop our
algorithm allows us to focus on and formalize each of the pruning rules
individually (Section \ref{sec:strategies}).  This aids understanding of the
algorithm and eases insertion of future improvements and better pruning rules.
\end{itemize}
\vspace*{-0.8em}

Section \ref{sec:scaling} shows the relevance of the large $k$ case; then, in
Section \ref{sec:trees}, we show that we can build a tree on the $k$ clusters,
and then a dual-tree algorithm \cite{curtin2013tree} can be used to efficiently
perform an exact single iteration of $k$-means clustering.  Section
\ref{sec:strategies} details the four pruning strategies used in our algorithm,
and Section \ref{sec:algorithm} introduces the algorithm itself.  Sections
\ref{sec:theory} and \ref{sec:empirical} show the theoretical and empirical
results for the algorithm, and finally Section \ref{sec:conclusion} concludes
the paper and paints directions for future improvements.

\vspace*{-0.2em}
\section{Scaling $k$-means}
\label{sec:scaling}
\vspace*{-0.1em}

Although the original publications on $k$-means only applied the algorithm to a
maximum dataset size of 760 points, the half-century of relentless progress
since then has seen dataset sizes scale into billions.  Due to its
simplicity, though, $k$-means has remained relevant, and is still applied in
many large-scale applications.

In cases where $N$ scales but $k$ remains small, a good choice of algorithm is a
sampling algorithm, which will return an approximate clustering.  One sampling
technique, coresets, can produce good clusterings for $n$ in the millions using
several hundred or a few thousand points \cite{coresets}.  However, for large
$k$, the number of samples required to produce good clusterings can become
prohibitive.

For large $k$, then, we turn to an alternative approach: accelerating exact
Lloyd iterations.  Existing techniques include the brute-force
implementation, the {\it blacklist} algorithm
\cite{pelleg1999accelerating}, Elkan's algorithm \yrcite{elkan2003using}, and
Hamerly's algorithm \yrcite{hamerly2010making}, as well as the recent Yinyang
$k$-means algorithm \cite{ding2015yinyang}.  The blacklist algorithm builds a
$kd$-tree on the dataset and, while the tree is traversed, blacklists individual
clusters that cannot be the closest cluster (the {\it owner}) of any descendant
points of a node.  Elkan's algorithm maintains an upper bound and a lower bound
on the distance between each point and centroid; Hamerly's algorithm is a
memory-efficient
simplification of this technique.  The Yinyang algorithm
organizes the centroids into groups of about 10 (depending on algorithm
parameters) using 5 iterations of $k$-means on the centroids followed by a
single iteration of standard $k$-means on the points.  Once groups are built,
the Yinyang algorithm attempts to prune groups of centroids at a time using
rules similar to Elkan and Hamerly's algorithms.

Of these algorithms, only Yinyang $k$-means considers centroids in groups at
all, but it does not consider points in groups.  On the other hand, the
blacklist algorithm is the only algorithm that builds a tree on the points and
is able to assign multiple points to a single cluster at once.  So, although
each algorithm has its own useful region, none of the four we have considered
here are particularly suited to the case of large $N$ {\bf and} large $k$.

Table \ref{tab:runtimes} shows setup costs,
worst-case per-iteration runtimes, and memory usage of each of these algorithms
as well as the proposed dual-tree algorithm\footnote{The dual-tree algorithm
worst-case runtime bound also depends on some assumptions on dataset-dependent
constants.  This is detailed further in Section \ref{sec:theory}.}.  The
expected runtime of the blacklist algorithm is, under some assumptions,
$O(k + k \log N + N)$ per iteration.  The expected runtime of Hamerly's and
Elkan's algorithm is $O(k^2 + \alpha N)$ time, where $\alpha$ is the expected
number of clusters visited by each point (in both Elkan and Hamerly's results,
$\alpha$ seems to be small).

However, none of these algorithms are specifically tailored to the large $k$
case, and the large $k$ case is common.  Pelleg and Moore
\yrcite{pelleg1999accelerating} report several hundred clusters in a subset of
800k objects from the SDSS dataset.  Clusterings for $n$-body simulations on
astronomical data often involve several thousand clusters
\cite{kwon2010scalable}.  Csurka
et~al. \yrcite{csurka} extract vocabularies from image sets using $k$-means with
$k \sim 1000$.  Coates et~al. \yrcite{coates} show that $k$-means can work
surprisingly well for unsupervised feature learning for images, using $k$ as
large as 4000 on 50000 images.  Also, in text mining, datasets can have up to
18000 unique labels \cite{bengio2010label}.  Can and Ozkarahan
\yrcite{can1990concepts} suggest that the number of clusters in text data is
directly related to the size of the vocabulary, suggesting $k \sim mN/t$ where
$m$ is the vocabulary size, $n$ is the number of documents, and $t$ is the
number of nonzero entries in the term matrix.
Thus, it is important to have an algorithm with favorable scaling properties for
both large $k$ and $N$.

\vspace*{-0.4em}
\section{Tree-based algorithms}
\label{sec:trees}
\vspace*{-0.2em}

The blacklist algorithm is an example of a {\it single-tree algorithm}: one tree
(the {\it reference tree}) is built on the dataset, and then that tree is
traversed.  This approach is applicable to a surprising variety of
other problems, too \cite{bentley1975multidimensional, moore1998very,
curtin2013fast}.  Following the blacklist algorithm, then, it is only
natural to build a tree on the data points.  Tree-building is (generally) a
one-time $O(N \log N)$ cost and for large $N$ or $k$, the cost of tree
building is often negligible compared to the time it takes to perform the
clustering.

\setlength{\textfloatsep}{0.4em}
\begin{figure*}[t]
\centering
\begin{subfigure}[b]{0.31\textwidth}
  \begin{tikzpicture}
    \filldraw [lightgray!60!blue] (0.0, 0.0) circle (0.6) { };
    \draw [thin] (0.0, 0.0) circle (0.6) { };
    \node [ ] at (0.0, 0.0) { $\mathscr{N}_q$ };

    \filldraw [lightgray!60!red] (-0.6, 0.2) circle (0.4) { };
    \draw [thin] (-0.6, 0.2) circle (0.4) { };
    \node [ ] at (-0.6, 0.2) { $\mathscr{N}_{r2}$ };

    \filldraw [lightgray!60!red] (2.1, 0.1) circle (0.4) { };
    \draw [thin] (2.1, 0.1) circle (0.4) { };
    \node [ ] at (2.1, 0.1) { $\mathscr{N}_r$ };

    \draw [black,dashed,domain=-30:30] plot ({1.63246*cos(\x)},
{1.63246*sin(\x)});
    \draw [black,dashed,domain=150:210] plot ({1.63246*cos(\x)},
{1.63246*sin(\x)});

    \draw (0.6, 0.0) -- (1.63246, 0.0) { };
    \draw (0.6, 0.0) -- (0.7, 0.1) { };
    \draw (0.6, 0.0) -- (0.7, -0.1) { };
    \draw (1.63246, 0.0) -- (1.53246, 0.1) { };
    \draw (1.63246, 0.0) -- (1.53246, -0.1) { };
    \node [ ] at (1.1, 0.3) { $\scriptstyle{\ub(\N_q)}$ };
  \end{tikzpicture}
  \caption{$\mathscr{N}_r$ can be pruned.}
  \label{fig:prune-1}
\end{subfigure}
\begin{subfigure}[b]{0.31\textwidth}
  \begin{tikzpicture}
    \draw [gray,dashed] (-0.65, -0.5) -- (-0.65, 1.1) { };

    \filldraw [lightgray!60!blue] (0.0, 0.0) circle (0.05) { };
    \draw [thin] (0.0, 0.0) circle (0.05) { };
    \node [ ] at (-0.26, 0.0) { $p_q$ };

    \filldraw [lightgray!60!red] (0.4, 0.2) circle (0.05) { };
    \draw [thin] (0.4, 0.2) circle (0.05) { };
    \node [ ] at (0.3, 0.42) { $c_j$ };

    \draw [gray] (0.44, 0.23) -- (1.2, 0.6) { };
    \draw [gray] (0.44, 0.23) -- (0.47, 0.3) { };
    \draw [gray] (0.44, 0.23) -- (0.51, 0.2) { };
    \draw [gray] (1.2, 0.6) -- (1.17, 0.53) { };
    \draw [gray] (1.2, 0.6) -- (1.13, 0.63) { };
    \node [ ] at (0.9, 0.1) { $m_j$ };

    \draw [black,dashed,domain=-30:50] plot ({1.3416*cos(\x)},
{1.3416*sin(\x)});

    \draw [gray] (0.045, -0.005) -- (1.3212, -0.23297) { };
    \draw [gray] (0.045, -0.005) -- (0.1, 0.03) { };
    \draw [gray] (0.045, -0.005) -- (0.095, -0.065) { };
    \draw [gray] (1.3212, -0.23297) -- (1.25, -0.26) { };
    \draw [gray] (1.3212, -0.23297) -- (1.27, -0.17) { };
    \node [ ] at (0.5, -0.4) { $\scriptstyle{\ub(p_q) + m_j}$ };

    \filldraw [lightgray!60!red] (2.4, 0.3) circle (0.05) { };
    \draw [thin] (2.4, 0.3) circle (0.05) { };
    \node [ ] at (2.4, 0.5) { $c_k$ };

    \draw [gray] (2.45, 0.3) -- (3.3, 0.3) { };
    \draw [gray] (2.45, 0.3) -- (2.52, 0.27) { };
    \draw [gray] (2.45, 0.3) -- (2.52, 0.33) { };
    \draw [gray] (3.3, 0.3) -- (3.23, 0.27) { };
    \draw [gray] (3.3, 0.3) -- (3.23, 0.33) { };
    \node [ ] at (2.8, 0.1) { $\scriptstyle{\min_k m_k}$ };

    \draw [black,dotted] (2.4, 0.3) circle (0.9) { };

  \end{tikzpicture}
  \caption{$p_q$'s owner cannot change.}
  \label{fig:prune-2}
\end{subfigure}
\begin{subfigure}[b]{0.31\textwidth}
  \begin{tikzpicture}
    \draw [gray,dashed] (-0.9, -0.5) -- (-0.9, 1.1) { };

    \filldraw [lightgray!60!blue] (0.0, 0.0) circle (0.05) { };
    \draw [thin] (0.0, 0.0) circle (0.05) { };
    \node [ ] at (-0.26, 0.0) { $p_q$ };

    \filldraw [lightgray!60!red] (0.4, 0.2) circle (0.05) { };
    \draw [thin] (0.4, 0.2) circle (0.05) { };
    \node [ ] at (0.3, 0.42) { $c_j$ };

    \draw [gray] (0.44, 0.23) -- (1.2, 0.6) { };
    \draw [gray] (0.44, 0.23) -- (0.47, 0.3) { };
    \draw [gray] (0.44, 0.23) -- (0.51, 0.2) { };
    \draw [gray] (1.2, 0.6) -- (1.17, 0.53) { };
    \draw [gray] (1.2, 0.6) -- (1.13, 0.63) { };
    \node [ ] at (0.9, 0.1) { $m_j$ };

    \draw [black,dashed,domain=-30:50] plot ({1.3416*cos(\x)},
{1.3416*sin(\x)});

    \draw [gray] (0.045, -0.005) -- (1.3212, -0.23297) { };
    \draw [gray] (0.045, -0.005) -- (0.1, 0.03) { };
    \draw [gray] (0.045, -0.005) -- (0.095, -0.065) { };
    \draw [gray] (1.3212, -0.23297) -- (1.25, -0.26) { };
    \draw [gray] (1.3212, -0.23297) -- (1.27, -0.17) { };
    \node [ ] at (0.5, -0.4) { $\scriptstyle{\ub(p_q) + m_j}$ };

    \filldraw [lightgray!60!red] (2.0, 0.3) circle (0.05) { };
    \draw [thin] (2.0, 0.3) circle (0.05) { };
    \node [ ] at (2.0, 0.5) { $c_k$ };

    \draw [gray] (2.05, 0.3) -- (2.9, 0.3) { };
    \draw [gray] (2.05, 0.3) -- (2.12, 0.33) { };
    \draw [gray] (2.05, 0.3) -- (2.12, 0.27) { };
    \draw [gray] (2.9, 0.3) -- (2.83, 0.33) { };
    \draw [gray] (2.9, 0.3) -- (2.83, 0.27) { };
    \node [ ] at (2.4, 0.1) { $\scriptstyle{\min_k m_k}$ };

    \draw [black,dotted] (2.0, 0.3) circle (0.9) { };

  \end{tikzpicture}
\caption{$p_q$'s owner can change.}
\label{fig:prune-3}
\end{subfigure}
\vspace*{-0.7em}
\caption{Different pruning situations.}
\label{fig:prune_one}
\vspace*{-1.0em}
\end{figure*}
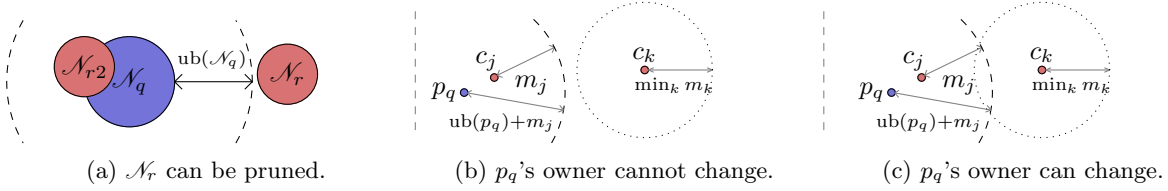

The speedup of the blacklist algorithm comes from the hierarchical nature of
trees: during the algorithm, we may rule out a cluster centroid for {\it many
points at once}.  The same reason is responsible for the impressive speedups
obtained for other single-tree algorithms, such as nearest neighbor search
\cite{bentley1975multidimensional, liu2004investigation}.  But
for nearest neighbor search, the nearest neighbor is often required not just for
a query point but instead a {\it query set}.  This observation
motivated the development of {\it dual-tree algorithms}, which also build a tree
on the query set (the {\it query tree}) in order to share work across query
points.  Both trees are recursed in such a way that combinations of query
nodes and reference nodes are visited.  Pruning criteria are applied to
these node combinations, and if a combination may be pruned, then the
recursion does not continue in that direction.

This approach
is applicable to $k$-means with large $k$: we may build a tree on the
$k$ cluster centroids, as well as a tree on the data points, and then we may
rule out {\it many} centroids for {\it many} points at once.

A recent result generalizes the class of dual-tree
algorithms, simplifying their expression and development
\cite{curtin2013tree}.  Any dual-tree algorithm can be decomposed into three
parts: a type of space tree, a pruning dual-tree traversal, and a point-to-point
\texttt{BaseCase()} function and node-to-node \texttt{Score()} function that
determines when pruning is possible.  Precise definitions and details of the
abstraction are given by \citet{curtin2013tree}, but for our purposes, this means
that we can describe a dual-tree $k$-means algorithm entirely with a
straightforward \texttt{BaseCase()} function and \texttt{Score()} function.  Any
tree and any traversal can then be used to create a working dual-tree algorithm.


The two types of trees we will explicitly consider in this paper are the
$kd$-tree and the cover tree \cite{langford2006}, but it should be remembered
that the algorithm as provided is sufficiently general to work with any other
type of tree.  Therefore, we standardize notation for trees: a tree is denoted
with $\mathscr{T}$, and a node in the tree is denoted by $\mathscr{N}$.  Each
node in a tree may have children; the set of children of $\mathscr{N}_i$ is
denoted $\mathscr{C}_i$.  In addition, each node may hold some points; this set
of points is denoted $\mathscr{P}_i$.  Lastly, the set of {\it descendant}
points of a node $\mathscr{N}_i$ is denoted $\mathscr{D}^p_i$.  The descendant
points are all points held by descendant nodes, and it is important to note that
the set $\mathscr{P}_i$ is {\it not} equivalent to $\mathscr{D}^p_i$.  This
notation is taken from \citet{curtin2013tree} and is detailed more
comprehensively there.  Lastly, we say that a centroid $c$ {\it owns} a point
$p$ if $c$ is the closest centroid to $p$.

\vspace*{-0.5em}
\section{Pruning strategies}
\label{sec:strategies}
\vspace*{-0.2em}

All of the existing accelerated $k$-means algorithms operate by avoiding
unnecessary work via the use of pruning strategies.  Thus, we will pursue four
pruning strategies, each based on or related to earlier work
\cite{pelleg1999accelerating, elkan2003using, hamerly2010making}.

These pruning strategies are meant to be used during the dual-tree traversal,
for which we have built a query tree $\mathscr{T}_q$ on the points and
a reference tree $\mathscr{T}_r$ on the centroids.  Therefore, these pruning
strategies consider not just combinations of single points and centroid
$p_q$ and $c_i$, but the combination of sets of points and sets of centroids,
represented by a query tree node $\N_q$ and a centroid tree node $\N_r$.  This
allows us to prune many centroids for many points simultaneously.


{\bf Strategy one.} When visiting a particular combination $(\N_q, \N_r)$
(with $\N_q$ holding points in the dataset and $\N_r$ holding
centroids), the combination should be pruned if every descendant centroid in
$\N_r$ can be shown to
own none of the points in $\N_q$.  If we have cached an upper bound $\ub(\N_q)$
on the distance between any descendant point of $\N_q$ and its nearest cluster
centroid that satisfies

\vspace*{-1.0em}
\begin{equation}
\ub(\N_q) \ge \max_{p_q \in \mathscr{D}^p_q} d(p_q,
c_q)
\end{equation}
\vspace*{-1.0em}

\noindent where $c_q$ is the cluster centroid nearest to point $p_q$, then the
node $\N_r$ can contain no centroids that own any descendant points of $\N_q$ if

\vspace*{-1.0em}
\begin{equation}
d_{\min}(\N_q, \N_r) > \ub(\N_q).
\label{eqn:prune}
\end{equation}
\vspace*{-1.5em}

This relation bears similarity to the pruning rules for nearest neighbor search
\cite{curtin2013tree} and max-kernel search \cite{curtin2014dual}.  Figure
\ref{fig:prune-1} shows a situation where $\N_r$ can be pruned; in this case,
ball-shaped tree nodes are used, and the upper bound $\ub(\N_q)$ is set to
$d_{\max}(\N_q, \N_{r2})$.

{\bf Strategy two.} The recursion down a particular branch of the query tree
should
terminate early if we can determine that only one cluster can possibly own all
of the descendant points of that branch.  This is related to the first strategy.
If we have been caching the number of pruned centroids (call this
$\pruned(\N_q)$), as well as the identity of any
arbitrary non-pruned centroid (call this $\closest(\N_q)$), then if
$\pruned(\N_q) = k - 1$, we may conclude that the
centroid $\closest(\N_q)$ is the owner of all descendant
points of $\N_q$, and there is no need for further recursion in
$\N_q$.

{\bf Strategy three.} The traversal should not visit nodes whose owner could not
have possibly changed between iterations; that is, the tree should be coalesced
to include only nodes whose owners may have changed.

There are two easy ways to use the triangle inequality to show that the owner of
a point cannot change between iterations.  Figures \ref{fig:prune-2} and
\ref{fig:prune-3} show the
first: we have a point $p_q$ with owner $c_j$ and second-closest
centroid $c_k$.  Between iterations, each centroid will move when it is
recalculated; define the distance that centroid $c_i$ has moved as
$m_i$.  Then we bound the distances for the next
iteration: $d(p_q, c_j) + m_j$ is an upper bound on the distance from $p_q$
to its owner next iteration, and $d(p_q, c_k) - \max_i m_i$ is a lower bound on
the distance from $p_q$ to its second closest centroid next iteration.  We
may use these bounds to conclude that if

\vspace*{-1.2em}
\begin{equation}
d(p_q, c_j) + m_j < d(p_q, c_k) - \max_i m_i,
\end{equation}
\vspace*{-1.6em}

\noindent then the owner of $p_q$ next iteration must be $c_j$.  Generalizing
from individual points $p_q$ to tree nodes $\N_q$ is easy.
This pruning strategy can only be used when all descendant points of
$\N_q$ are owned by a single centroid, and in order to perform the
prune, we need to establish a lower bound on the distance between any
descendant point of the node $\N_q$ and the second closest centroid.
Call this bound $\lb(\N_q)$.  Remember that
$\ub(\N_q)$ provides an upper bound on the distance
between any descendant point of $\N_q$ and its nearest centroid.  Then,
if all descendant points of $\N_q$ are owned by some cluster $c_j$ in
one iteration, and

\vspace*{-1.4em}
\begin{equation}
\ub(\N_q) + m_j < \lb(\N_q) -
\max_i m_i,
\label{eqn:static-1}
\end{equation}
\vspace*{-1.4em}

\noindent then $\N_q$ is owned by cluster $c_j$ in the next iteration.
Implementationally, it is convenient to have $\lb(\N_q)$ store a
lower bound on the distance between any descendant point of $\N_q$ and
the nearest pruned centroid.  Then, if $\N_r$ is entirely owned by one
cluster, all other centroids are pruned, and $\lb(\N_q)$
holds the necessary lower bound for pruning according to the rule above.

The second way to use the triangle inequality to show that an owner cannot
change depends on the distances between centroids.  Suppose that $p_q$ is
owned by $c_j$ at the current iteration; then, if

\vspace*{-1.3em}
\begin{equation}
d(p_q, c_j) - m_j < 2 \left( \min_{c_i \in C, c_i \ne c_j} d(c_i, c_j) \right)
\end{equation}
\vspace*{-1.3em}

\noindent then $c_j$ will own $p_q$ next iteration \cite{elkan2003using}.  We
may adapt this rule to tree nodes $\N_q$ in the same way as the previous rule;
if $\N_q$ is owned by cluster $c_j$ during this iteration and

\vspace*{-1.3em}
\begin{equation}
\ub(\N_q) + m_j < 2 \left( \min_{c_i \in C, c_i \ne c_j}
d(c_i, c_j) \right)
\label{eqn:static-2}
\end{equation}
\vspace*{-1.3em}

\noindent then $\N_q$ is owned by cluster $c_j$ in the next iteration.
Note that the above rules do work with individual points $p_q$ instead of nodes
$\N_q$ if we have a valid upper bound $\ub(p_q)$ and a
valid lower bound $\lb(p_q)$.  Any nodes or points that satisfy
the above conditions do not need to be visited during the next iteration, and
can be removed from the tree for the next iteration.

{\bf Strategy four.} The traversal should use bounding information from previous
iterations; for instance, $\ub(\N_q)$ should not be reset
to $\infty$ at the beginning of each iteration.  Between iterations, we may
update $\ub(\N_q)$, $\ub(p_q)$,
$\lb(\N_q)$, and $\lb(p_q)$ according to
the following rules:

\vspace*{-1.5em}
\begin{eqnarray}
\ub(\N_q) &\gets&
  \begin{cases}
    \ub(\N_q) + m_j & \text{if } \N_q \text{ is}\\
\multicolumn{2}{l}{\text{\ \ \ \ owned by a single cluster $c_j$}}
\\
    \ub(\N_q) + \max_i m_i & \text{if } \N_q \text{ is}\\
\multicolumn{2}{l}{\text{\ \ \ \ not owned by a single cluster},}
  \end{cases} \label{eqn:special} \\
\ub(p_q) &\gets& \ub(p_q) + m_j, \\
\lb(\N_q) &\gets& \lb(\N_q) -
\max_i m_i, \\
\lb(p_q) &\gets& \lb(p_q) - \max_i m_i.
\end{eqnarray}
\vspace*{-1.0em}

Special handling is required when descendant points of $\N_q$
are not owned by a single centroid (Equation \ref{eqn:special}).  It is also
true that for a child node $\N_c$ of $\N_q$, $\ub(\N_q)$ is a valid upper bound
for $\N_c$ and $\lb(\N_q)$ is a valid lower bound for $\N_c$: that is, the upper
and lower bounds may be taken from a parent, and they are still valid.

\vspace*{-0.6em}
\section{The dual-tree $k$-means algorithm}
\label{sec:algorithm}
\vspace*{-0.3em}

These four pruning strategies lead to a high-level $k$-means algorithm,
described in Algorithm \ref{alg:high_level}.  During the course of this
algorithm, to implement each of our pruning strategies, we will need to maintain
the following quantities:

\vspace*{-1.0em}
\begin{itemize} \itemsep -1.5pt
  \item $\ub(\N_q)$: an upper bound on the distance
between any descendant point of a node $\N_q$ and the nearest centroid
to that point.
  \item $\lb(\N_q)$: a lower bound on the distance
between any descendant point of a node $\N_q$ and the nearest pruned
centroid.
  \item $\pruned(\N_q)$: the number of centroids pruned
during traversal for $\N_q$.
  \item $\closest(\N_q)$: if $\pruned(\N_q) = k - 1$, this
holds the owner of all descendant points of $\N_q$.
  \item $\canchange(\N_q)$: whether or not
$\N_q$ can change owners next iteration.
  \item $\ub(p_q)$: an upper bound on the distance between point
$p_q$ and its nearest centroid.
  \item $\lb(p_q)$: a lower bound on the distance between point
$p_q$ and its second nearest centroid.
  \item $\closest(p_q)$: the closest centroid to $p_q$ (this is
also the owner of $p_q$).
  \item $\canchange(p_q)$: whether or not $p_q$ can change owners
next iteration.
\end{itemize}
\vspace*{-0.8em}

At the beginning of the algorithm, each upper bound is initialized to $\infty$,
each lower bound is initialized to $\infty$, $\pruned(\cdot)$
is initialized to $0$ for each node, and
$\closest(\cdot)$ is initialized to an invalid centroid for each
node and point.  $\canchange(\cdot)$ is set to {\tt
true} for each node and point.  Thus line
6 does nothing on the first iteration.

\setlength{\textfloatsep}{0.4em}
\begin{algorithm}[t!]
\begin{algorithmic}[1]
  \STATE {\bf Input:} dataset $S \in \mathcal{R}^{N \times d}$, initial
centroids $C \in \mathcal{R}^{k \times d}$.
  \STATE {\bf Output:} converged centroids $C$.
  \medskip
  \STATE $\mathscr{T} \gets$ a tree built on $S$
  \WHILE{centroids $C$ not converged}
    \STATE \COMMENT{Remove nodes in the tree if possible.}
    \STATE $\mathscr{T} \gets \mathtt{CoalesceNodes(}\mathscr{T}\mathtt{)}$
    \STATE $\mathscr{T}_c \gets$ a tree built on $C$
    \medskip
    \STATE \COMMENT{Call dual-tree algorithm.}
    \STATE Perform a dual-tree recursion with $\mathscr{T}$, $\mathscr{T}_c$,
\texttt{BaseCase()}, and \texttt{Score()}.
    \medskip
    \STATE \COMMENT{Restore the tree to its non-coalesced form.}
    \STATE $\mathscr{T} \gets \mathtt{DecoalesceNodes(\mathscr{T})}$
    \medskip
    \STATE \COMMENT{Update centroids and bounding information.}
    \STATE $C \gets \mathtt{UpdateCentroids(}\mathscr{T}\mathtt{)}$
    \STATE $\mathscr{T} \gets \mathtt{UpdateTree(}\mathscr{T}\mathtt{)}$
  \ENDWHILE
  \STATE {\bf return} $C$
\end{algorithmic}
\caption{High-level outline of dual-tree $k$-means.}
\label{alg:high_level}
\end{algorithm}

First, consider the dual-tree algorithm called on line
9.  As detailed earlier, we can describe a dual-tree
algorithm as a combination of tree type, traversal, and point-to-point
\texttt{BaseCase()} and node-to-node \texttt{Score()} functions.  Thus, we
need only present \texttt{BaseCase()} (Algorithm \ref{alg:base_case}) and
\texttt{Score()} (Algorithm \ref{alg:score})\footnote{In these algorithms, we
assume that any point present in a node $\N_i$ will also be present in at least
one child $\N_c \in \mathscr{C}_i$.  It is possible to fully
generalize to any tree type, but the exposition is significantly more complex,
and our assumption covers most standard tree types anyway.}.

The \texttt{BaseCase()} function is simple: given a point $p_q$ and a
centroid $c_r$, the distance $d(p_q, c_r)$ is calculated; $\ub(p_q)$,
$\lb(p_q)$, and $\closest(p_q)$ are updated if needed.

\texttt{Score()} is more complex.  The first stanza (lines 4--6) takes the
values of $\pruned(\cdot)$ and $\lb(\cdot)$ from the parent node of $\N_q$; this
is necessary to prevent $\pruned(\cdot)$ from undercounting.  Next, we prune if
the owner of $\N_q$ is already
known (line 7).  If the minimum distance between any descendant point of $\N_q$
and any descendant centroid of $\N_r$ is greater than $\ub(\N_q)$,
then we may prune the combination (line 16).  In that case we may also improve
the lower bound (line 14).  Note the special handling in line 15: our definition
of tree allows points to be held in more than one node; thus, we must avoid
double-counting clusters that we prune.\footnote{For trees like the $kd$-tree
and the metric tree, which do not hold points in more than one node, no special
handling is required: we will never prune a cluster twice for a given query node
$\N_q$.}.  If the node combination cannot be pruned in this way, an attempt is
made to update the upper bound (lines 17--20).  Instead of using $d_{\max}(\N_q,
\N_r)$, we may use a tighter upper bound: select any
descendant centroid $c$ from $\N_r$ and use $d_{\max}(\N_q, c)$.  This still
provides a valid upper bound, and in practice is generally smaller than
$d_{\max}(\N_q, \N_r)$.  We simply set $\closest(\N_q)$ to $c$ (line 20);
$\closest(\N_q)$ only holds the owner of $\N_q$ if all centroids
except one are pruned---in which case the owner {\it must} be $c$.

\setlength{\textfloatsep}{0.4em}
\begin{algorithm}[t!]
\begin{algorithmic}[1]
  \STATE {\bf Input:} query point $p_q$, reference centroid $c_r$
  \STATE {\bf Output:} distance between $p_q$ and $c_r$
  \medskip
  \IF{$d(p_q, c_r) < \ub(p_q)$}
    \STATE $\lb(p_q) \gets \ub(p_q)$
    \STATE $\ub(p_q) \gets d(p_q, c_r)$
    \STATE $\closest(p_q) \gets c_r$
  \ELSIF{$d(p_q, c_r) < \lb(p_q)$}
    \STATE $\lb(p_q) \gets d(p_q, c_r)$
  \ENDIF
  \medskip
  \STATE {\bf return} $d(p_q, c_r)$
\end{algorithmic}
\caption{\texttt{BaseCase()} for dual-tree $k$-means.}
\label{alg:base_case}
\end{algorithm}

\begin{algorithm}[t!]
\begin{algorithmic}[1]
  \STATE {\bf Input:} query node $\N_q$, reference node $\N_r$
  \STATE {\bf Output:} score for node combination $(\N_q,
\N_r)$, or $\infty$ if the combination can be pruned
  \medskip
  \STATE \COMMENT{Update the number of pruned nodes, if needed.}
  \IF{$\N_q$ not yet visited and is not the root node}
    \STATE $\pruned(\N_q) \gets
\parent(\N_q)$
    \STATE $\lb(\N_q) \gets
\lb(\parent(\N_q))$
  \ENDIF
  \STATE{{\bf if} $\pruned(\N_q) = k - 1$ {\bf then return} $\infty$}
  \medskip
  \STATE $s \gets d_{\min}(\N_q, \N_r)$
  \STATE $c \gets \mathrm{any\ descendant\ cluster\ centroid\ of } \N_r$
  \IF{$d_{\min}(\N_q, \N_r) >
\ub(\N_q)$}
    \STATE \COMMENT{This cluster node owns no descendant points.}
    \IF{$d_{\min}(\N_q, \N_r) <
\lb(\N_q)$}
      \STATE \COMMENT{Improve the lower bound for pruned nodes.}
      \STATE $\lb(\N_q) \gets d_{\min}(\N_q,
\N_r)$
    \ENDIF
    \STATE $\pruned(\N_q) \mathrel{+}= |\mathscr{D}^p_r \setminus \{ \textrm{clusters
not pruned} \}|$
    \STATE $s \gets \infty$
  \medskip
  \ELSIF{$d_{\max}(\N_q, c) <
\ub(\N_q)$}
    \STATE \COMMENT{We may improve the upper bound.}
    \STATE $\ub(\N_q) \gets d_{\max}(\N_q,
\N_r)$
    \STATE $\closest(\N_q) \gets c$
  \ENDIF
  \medskip
  \STATE \COMMENT{Check if all clusters (except one) are pruned.}
  \STATE {\bf if} $\pruned(\N_q) = k - 1$ {\bf then return} $\infty$
  \medskip
  \STATE {\bf return} $s$
\end{algorithmic}
\caption{\texttt{Score()} for dual-tree $k$-means.}
\label{alg:score}
\end{algorithm}

\begin{algorithm}[t!]
\begin{algorithmic}[1]
  \STATE {\bf Input:} tree $\mathscr{T}$ built on dataset $S$
  \STATE {\bf Output:} new centroids $C$
  \medskip
  \STATE $C := \{ c_0, \ldots, c_{k - 1} \} \gets \bm{0}^{k \times d}$; \ $n =
\bm{0}^k$
  \medskip
  \STATE \COMMENT{$s$ is a stack.}
  \STATE $s \gets \{ \operatorname{root}(\mathscr{T}) \}$
  \WHILE{$|s| > 0$}
    \STATE $\N_i \gets s\mathtt{.pop()}$
    \IF{$\pruned(\N_i) = k - 1$}
      \STATE \COMMENT{The node is entirely owned by a cluster.}
      \STATE $j \gets \mathrm{index\ of } \closest(\N_i)$
      \STATE $c_j \gets c_j + |\mathscr{D}^p_i|
\operatorname{centroid}(\N_i)$
      \STATE $n_j \gets n_j + |\mathscr{D}^p_i|$
    \ELSE
      \STATE \COMMENT{The node is not entirely owned by a cluster.}
      \STATE {{\bf if} $|\mathscr{C}_i| > 0$ {\bf then}
$s\mathtt{.push(}\mathscr{C}_i\mathtt{)}$}
      \STATE {\bf else}
        \STATE {\ \ \ \ {\bf for} $p_i \in \mathscr{P}_i$ not yet considered}
          \STATE \ \ \ \ \ \ \ $j \gets \mathrm{index\ of } \closest(p_i)$
          \STATE \ \ \ \ \ \ \ $c_j \gets c_j + p_i$; \ \ $n_j \gets n_j + 1$
    \ENDIF
  \ENDWHILE
  \medskip
  \STATE{{\bf for} $c_i \in C${\bf, if} $n_i > 0$ {\bf then} $c_i \gets c_i /
n_i$}
  \STATE {\bf return} $C$
\end{algorithmic}
\caption{\texttt{UpdateCentroids()}.}
\label{alg:update_centroids}
\end{algorithm}

Thus, at the end of the dual-tree algorithm, we know the owner of every node (if
it exists) via $\closest(\cdot)$ and $\pruned(\cdot)$, and we know the owner of
every point via $\closest(\cdot)$.  A simple
algorithm to do this is given here as Algorithm \ref{alg:update_centroids}
(\texttt{UpdateCentroids()}); it
is a depth-first recursion through the tree that terminates a branch when a node
is owned by a single cluster.


Next is updating the bounds in the tree and determining if nodes and
points can change owners next iteration; this work is encapsulated in the
\texttt{UpdateTree()} algorithm, which is an implementation of strategies 3 and
4 (see the appendix for details).  Once
\texttt{UpdateTree()} sets the correct value of $\canchange(\cdot)$ for every
point and node, we coalesce the tree for the next iteration with the
\texttt{CoalesceTree()} function.  Coalescing the tree is straightforward:
we simply remove any nodes from the tree where $\canchange(\cdot)$
is \texttt{false}.  This leaves a smaller tree with no nodes where
$\canchange(\cdot)$ is \texttt{false}.
Decoalescing the tree (\texttt{DecoalesceTree()}) is done by restoring
the tree to its original state.  See the appendix for more details.

\vspace*{-0.7em}
\section{Theoretical results}
\label{sec:theory}
\vspace*{-0.4em}

Space constraints allow us to only provide proof sketches for the first two
theorems here.  Detailed proofs are given in the appendix.

\begin{thm}
A single iteration of dual-tree $k$-means as given in Algorithm
\ref{alg:high_level} will produce exactly the same results as the
brute-force $O(kN)$ implementation.
\end{thm}
\vspace*{-1.7em}
\begin{proof}
(Sketch.)  First, we show that the dual-tree algorithm (line 9) produces correct
results for $\ub(\cdot)$, $\lb(\cdot)$, $\pruned(\cdot)$, and $\closest(\cdot)$
for every point and node.  Next, we show that \texttt{UpdateTree()} maintains
the correctness of those four quantities and only marks $\canchange(\cdot)$ to
\texttt{false} when the node or point truly cannot change owner.  Next, it is
easily shown that \texttt{CoalesceTree()} and \texttt{DecoalesceTree()} do not
affect the results of the dual-tree algorithm because the only nodes and points
removed are those where $\canchange(\cdot) = \mathtt{false}$.  Lastly, we show
that \texttt{UpdateCentroids()} produces centroids correctly.
\end{proof}
\vspace*{-1.0em}

Next, we consider the runtime of the algorithm.  Our results are with respect to
the {\it expansion constant} $c_k$ of the centroids \cite{langford2006}, which
is a measure of intrinsic dimension.  $c_{qk}$ is a related quantity:
the largest expansion constant of $C$ plus any point in the dataset.  Our
results also depend on the imbalance of the tree $i_t(\mathscr{T})$, which in
practice generally scales linearly in $N$ \cite{curtin2015plug}.  As with the
other theoretical results, more detail on each of these quantities is available
in the appendix.

\begin{thm}
When cover trees are used, a single iteration of dual-tree $k$-means as in
Algorithm \ref{alg:high_level} can be performed in $O(c_k^4 c_{qk}^5 (N +
i_t(\mathscr{T})) + c_k^9 k \log k)$ time.
\end{thm}
\vspace*{-1.0em}
\begin{proof}
(Sketch.)  Cover trees have $O(N)$ nodes \cite{langford2006}; because
\texttt{CoalesceTree()}, \texttt{DecoalesceTree()}, \texttt{UpdateCentroids()},
and \texttt{UpdateTree()} can be performed in one pass of the tree, these steps
may each be completed in $O(N)$ time.  Building a tree on the centroids takes
$O(c_k^6 k \log k)$ time, where $c_k$ is the expansion constant of the
centroids.  Recent results show that dual-tree algorithms that use the cover
tree may have their runtime easily bounded \cite{curtin2015plug}.  We may
observe that our pruning rules are at least as tight as nearest neighbor search;
this means that the dual-tree algorithm (line 11) may be performed in
$O(c_{kr}^9 (N + i_t(\mathscr{T})))$ time.  Also, we must perform nearest
neighbor search on the centroids, which costs $O(c_k^9 (k +
i_t(\mathscr{T_c})))$ time.  This gives a total per-iteration runtime of
$O(c_{kr}^9 (N + i_t(\mathscr{T})) + c_k^6 k \log k + c_k^9
i_t(\mathscr{T}_k))$.
\end{proof}

\vspace*{-1.0em}

This result holds intuitively.  By building a tree on the centroids, we are able
to prune many centroids at once, and as a result the amortized cost of finding
the nearest centroid to a point is $O(1)$.  This meshes with earlier theoretical
results \cite{langford2006, curtin2015plug, ram2009} and earlier empirical
results \cite{gray2003nonparametric, gray2001nbody} that suggest that an answer
can be obtained for a single query point in $O(1)$ time.  Note that this
worst-case bound depends on the intrinsic dimension (the expansion constant)
of the centroids, $c_k$, and the related quantity $c_{qk}$.  If the intrinsic
dimension of the centroids is low---that is, if the centroids are distributed
favorably---the dual-tree algorithm will be more efficient.

However, this bound is generally quite loose in practice.  First, runtime bounds
for cover trees are known to be loose \cite{curtin2015plug}.  Second, this
particular bound does not consider the effect of coalescing the tree.
In any given iteration, especially toward the end of the $k$-means
clustering, most points will have $\operatorname{canchange}(\cdot) =
\mathtt{false}$ and thus the coalesced tree
will be far smaller than the full tree built on all $N$ points.

\begin{thm}
Algorithm \ref{alg:high_level} uses no more than $O(N + k)$ memory when cover
trees are used.
\end{thm}
\vspace*{-1.0em}
\begin{proof}
This proof is straightforward.  A cover tree on $N$ points takes $O(N)$
space.  So the trees and associated bounds take $O(N)$ and $O(k)$ space.  Also,
the dataset and centroids take $O(N)$ and $O(k)$ space.
\end{proof}

\vspace*{-1.3em}
\section{Experiments}
\label{sec:empirical}
\vspace*{-0.3em}

\setlength{\textfloatsep}{1.2em}
\begin{table}[t!]
{\small
\begin{center}
\begin{tabular}{|c|c|c|c|c|}
\hline
 & & & \multicolumn{2}{|c|}{\bf tree build time} \\
{\bf Dataset} & $N$ & $d$ & $kd$-tree & cover tree \\
\hline
cloud & 2048 & 10 & 0.001s & 0.005s \\
cup98b & 95413 & 56 & 1.640s & 32.41s \\
birch3 & 100000 & 2 & 0.037s & 2.125s \\
phy & 150000 & 78 & 4.138s & 22.99s \\
power & 2075259 & 7 & 7.342s & 1388s \\
lcdm & 6000000 & 3 & 4.345s & 6214s \\
\hline
\end{tabular}
\end{center}
}
\vspace*{-1.0em}
\caption{Dataset information.}
\label{tab:datasets}
\end{table}

\begin{table*}
\begin{center}
\resizebox{\textwidth}{!}{
\begin{tabular}{|c|c|r|l|l|l|l|l|l|}
\hline
 & & & \multicolumn{6}{|c|}{\bf avg. per-iteration runtime (distance
calculations)} \\
{\bf dataset} & $k$ & {\bf iter.} & {\tt elkan}                 & {\tt hamerly}           & {\tt yinyang}               & {\tt blacklist}             & {\tt dualtree-kd}           & {\tt dualtree-ct}
\\
\hline
cloud         & 3   & 8           & 1.50e-4s (867)              & 1.11e-4s (1.01k)        & 1.11e-1s (2.00k)             & {\bf 4.68e-5s} (302)        & 1.27e-4s ({\bf 278})        & 2.77e-4s (443)    \\
cloud         & 10  & 14          & 2.09e-4s ({\bf 1.52k})      & 1.92e-4s (4.32k)        & 7.66e-2s (9.55k)             & {\bf 1.55e-4s} (2.02k)      & 3.69e-4s (1.72k)            & 5.36e-4s (2.90k)  \\
cloud         & 50  & 19          & 5.87e-4s ({\bf 2.57k})      & {\bf 5.30e-4s} (21.8k)  & 9.66e-3s (15.6k)            & 8.20e-4s (12.6k)            & 1.23e-3s (5.02k)            & 1.09e-3s (9.84k)  \\
\hline
cup98b        & 50  & 224         & 0.0445s ({\bf 25.9k})       & 0.0557s (962k)          & 0.0465s (313k)              & {\bf 0.0409s} (277k)        & 0.0955s (254k)              & 0.1089s (436k) \\
cup98b        & 250 & 168         & 0.1972s ({\bf 96.8k})       & 0.4448s (8.40M)         & {\bf 0.1417s} (898k)        & 0.2033s (1.36M)             & 0.4585s (1.38M)             & 0.3237s (2.73M)   \\
cup98b        & 750 & 116         & 1.1719s ({\bf 373k})        & 1.8778s (36.2M)         & {\bf 0.2653s} (1.26M)       & 0.6365s (4.11M)             & 1.2847s (4.16M)             & 0.8056s (81.4M)   \\
\hline
birch3        & 50  & 129         & 0.0194s ({\bf 24.2k})       & 0.0093s (566k)          & 0.0378s (399k)              & {\bf 0.0030s} (42.7k)        & 0.0082s (37.4k)             & 0.0378s (67.9k) \\
birch3        & 250 & 812         & 0.0895s ({\bf 42.8k})       & 0.0314s (2.59M)         & 0.0711s (239k)              & {\bf 0.0164s} (165k)         & 0.0183s (79.7k)             & 0.0485s (140k)    \\
birch3        & 750 & 373         & 0.3253s (292k)              & 0.0972s (8.58M)         & 0.1423s (476k)              & 0.0554s (450k)               & {\bf 0.02989s} ({\bf 126k}) & 0.0581s (235k)    \\
\hline
phy           & 50  & 34          & 0.0668s (82.3k)             & 0.1064s (1.38M)         & 0.1072s (808k)              & {\bf 0.0081s} ({\bf 33.0k})  & 0.02689s (67.8k)            & 0.0945s (188k)    \\
phy           & 250 & 38          & 0.1627s (121k)              & 0.4634s (6.83M)         & 0.2469s (2.39M)             & {\bf 0.0249s} (104k)         & 0.0398s ({\bf 90.4k})       & 0.1023s (168k)    \\
phy           & 750 & 35          & 0.7760s ({\bf 410k})        & 2.9192s (43.8M)         & 0.6418s (5.61M)             & {\bf 0.2478s} (1.19M)        & 0.2939s (1.10M)             & 0.3330s (1.84M)   \\
\hline
power         & 25  & 4           & 0.3872s (2.98M)             & 0.2880s (12.9M)         & 1.1257s (33.5M)             & {\bf 0.0301s} (216k)         & 0.0950s ({\bf 87.4k})       & 0.6658s (179k)    \\
power         & 250 & 101         & 2.6532s (425k)              & 0.1868s (7.83M)         & 1.2684s (10.3M)             & 0.1504s (1.13M)              & {\bf 0.1354s} ({\bf 192k})  & 0.6405s (263k)    \\
power         & 1000& 870         & {\it out of memory}         & 6.2407s (389M)          & 4.4261s (9.41M)             & 0.6657s (2.98M)              & {\bf 0.4115s} ({\bf 1.57M}) & 1.1799s (4.81M) \\
power         & 5000& 504         & {\it out of memory}         & 29.816s (1.87B)         & 22.7550s (58.6M)            & 4.1597s (11.7M)              & {\bf 1.0580s} ({\bf 3.85M}) & 1.7070s (12.3M)   \\
power        & 15000& 301         & {\it out of memory}         & 111.74s (6.99B)         & {\it out of memory}         & {\it out of memory}          & {\bf 2.3708s} ({\bf 8.65M}) & 2.9472s (30.9M)   \\
\hline
lcdm          & 500 & 507         & {\it out of memory}         & 6.4084s (536M)          & 8.8926s (44.5M)             & 0.9347s (4.20M)              & {\bf 0.7574s} ({\bf 3.68M}) & 2.9428s (7.03M) \\
lcdm          & 1000& 537         & {\it out of memory}         & 16.071s (1.31B)         & 18.004s (74.7M)             & 2.0345s (5.93M)              & {\bf 0.9827s} ({\bf 5.11M}) & 3.3482s (10.0M)   \\
lcdm          & 5000& 218         & {\it out of memory}         & 64.895s (5.38B)         & {\it out of memory}         & 12.909s (16.2M)              & {\bf 1.8972s} ({\bf 8.54M}) & 3.9110s (19.0M)   \\
lcdm          &20000& 108         & {\it out of memory}         & 298.55s (24.7B)         & {\it out of memory}         & {\it out of memory}          & {\bf 4.1911s} ({\bf 17.8M}) & 5.5771s (43.2M)   \\
\hline
\end{tabular}
}
\end{center}
\vspace*{-1.0em}
\caption{Empirical results for $k$-means.}
\label{tab:runtime}
\vspace*{-1.0em}
\end{table*}

The next thing to consider is the empirical performance of the algorithm.  We
use the publicly available \texttt{kmeans} program in {\bf mlpack}
\cite{mlpack2013}; in our experiments, we run it as follows:

\vspace*{-0.5em}
\begin{verbatim}
$ kmeans -i dataset.csv -I centroids.csv -c
    $k -v -e -a $algorithm
\end{verbatim}
\vspace*{-0.5em}

\noindent where \texttt{\$k} is the number of clusters and \texttt{\$algorithm}
is the algorithm to be used.  Each algorithm is implemented in C++.  For the
{\tt yinyang} algorithm, we use the authors' implementation.  We use a variety
of $k$ values on mostly real-world datasets; details are shown in Table
\ref{tab:datasets} \cite{uci, birch3, lcdm}.  The table also contains the time
taken to build a $kd$-tree (for \texttt{blacklist} and \texttt{dualtree-kd}) and
a cover tree (for \texttt{dualtree-ct}).  Cover trees are far more complex to
build than $kd$-trees; this explains the long cover tree build time.  Even so,
the tree only needs to be built once during the $k$-means run.  If results are
required for multiple values of $k$---such as in the X-means algorithm
\cite{pelleg2000x}---then the tree built on the points may be re-used.

Clusters were initialized using the Bradley-Fayyad refined start procedure
\yrcite{bradley1998refining}; however, this was too slow for the very large
datasets, so in those cases points were randomly sampled as the initial
centroids.  $k$-means was then run until convergence on each dataset.  These
simulations were performed on a modest consumer desktop with an
Intel i5 with 16GB RAM, using {\bf mlpack}'s benchmarking system
\cite{edel2014automatic}.

Average runtime per iteration results are shown in Table \ref{tab:runtime}.
The amount of work that is being pruned away is somewhat unclear from the
runtime results, because the \texttt{elkan} and \texttt{hamerly} algorithms
access points linearly and thus benefit from cache effects; this is not true of
the tree-based algorithms.  Therefore, the average number of distance
calculations per iteration are also included in the results.

It is immediately clear that for large datasets, \texttt{dualtree-kd} is
fastest, and \texttt{dualtree-ct} is almost as fast.
The \texttt{elkan} algorithm, because it
holds $kN$ bounds, is able to prune away a huge amount of work and is very fast
for small datasets; however,
maintaining all of these bounds becomes prohibitive with large $k$ and the
algorithm exhausts all
available memory.  The \texttt{blacklist} algorithm has the same issue: on the
largest datasets, with the largest $k$ values, the space required to maintain
all the blacklists is too much.  This is also true of the \texttt{yinyang}
algorithm, which must maintain bounds between each point and each group of
centroids.  For large $k$, this burden becomes too much and the algorithm fails.
The \texttt{hamerly} and dual-tree algorithms, on the other hand, are the
best-behaved with memory usage and do not have any issues with large $N$ or
large $k$; however, the \texttt{hamerly} algorithm is very slow on large
datasets because it is not able to prune many points at once.

Similar to the observations about the \texttt{blacklist} algorithm, the
tree-based approaches are less effective in higher dimensions
\cite{pelleg1999accelerating}.  This is an important point: the performance of
tree-based approaches suffer in high dimensions in part because the bound 
$d_{\min}(\cdot, \cdot)$ generally becomes looser as dimension increases.
This is partly because the volume of nodes in high dimensions is much higher;
consider that a ball has volume that is exponential in the dimension.

Even so, in our results, we see speedup in reasonable dimensions (for example,
the {\tt phy} dataset has 78 dimensions).  Further, because our algorithm is
tree-independent, we may use tree structures that are tailored to
high-dimensional data \cite{arya1998optimal}---including ones that
have not yet been developed.  From our results we believe
as a rule of thumb that the dual-tree $k$-means algorithm can be effective up to
a hundred dimensions or more.

Another clear observation is that when $k$ is scaled on a single dataset, the
\texttt{dualtree-kd} and \texttt{dualtree-ct} algorithms nearly always scale
better (in terms of runtime) than the other algorithms.  These results show that
our algorithm satisfies its original goals: to be able to scale effectively to
large $k$ and $N$.

\vspace*{-0.6em}
\section{Conclusion and future directions}
\label{sec:conclusion}
\vspace*{-0.2em}

Using four pruning strategies, we have developed a flexible,
tree-independent dual-tree $k$-means algorithm that is the best-performing
algorithm for large datasets and large $k$ in small-to-medium dimensions.  It
is theoretically favorable, has a small memory footprint, and may be used in
conjunction with initial point selection and approximation schemes for
 additional speedup.

There are still interesting future directions to pursue, though.  The first
direction is parallelism: because our dual-tree algorithm is agnostic to the
type of traversal used, we may use a parallel traversal \cite{curtin2013tree},
such as an adapted version of a recent parallel dual-tree algorithm
\cite{lee2012distributed}.  The second direction is kernel $k$-means and other
spectral clustering techniques: our algorithm may be merged with the
ideas of \citet{curtin2014dual} to perform kernel $k$-means.  The third
direction is theoretical.  Recently, more general notions of intrinsic
dimensionality have been proposed \cite{houle2013dimensionality,
amsaleg2015estimating}; these may enable tighter and more descriptive runtime
bounds.  Our work thus provides a useful and fast $k$-means algorithm and also
opens promising avenues to further accelerated clustering algorithms.

\nocite{ram2009rank}
\nocite{march2010euclidean}

\bibliographystyle{icml2016}
\bibliography{kmeans}

\appendix

\section{Supplementary material}

Unfortunately, space constraints prevent adequate explanation of each of the
points in the main paper.  This supplementary material is meant to clarify all
of the parts of the dual-tree $k$-means algorithm that space did not permit in
the main paper.

\subsection{Updating the tree}

In addition to updating the centroids, the bounding information contained within
the tree must be updated according to pruning strategies 3 and 4.
Unfortunately, this yields a particularly complex recursive algorithm, given in
Algorithm \ref{alg:update_tree}.

\begin{algorithm*}
\begin{algorithmic}[1]
  \STATE {\bf Input:} node $\N_i$, $\ub(\cdot)$, $\lb(\cdot)$,
$\pruned(\cdot)$, $\closest(\cdot)$, $\canchange(\cdot)$, centroid movements $m$
  \STATE {\bf Output:} updated $\ub(\cdot)$, $\lb(\cdot)$, $\pruned(\cdot)$,
$\canchange(\cdot)$
  \medskip
  \STATE $\canchange(\N_i) \gets \mathtt{true}$
  \IF{$\N_i$ has a parent and $\canchange(\parent(\N_i)) = \mathtt{false}$}
    \STATE \COMMENT{Use the parent's bounds.}
    \STATE $\closest(\N_i) \gets \closest(\parent(\N_i))$
    \STATE $j \gets \mathrm{index\ of } \closest(\N_i)$
    \STATE $\ub(\N_i) \gets \ub(\N_i) + m_j$
    \STATE $\lb(\N_i) \gets \lb(\N_i) + \max_i m_i$
    \STATE $\canchange(\N_i) \gets \mathtt{false}$
  \ELSIF{$\pruned(\N_i) = k - 1$}
    \STATE \COMMENT{$\N_i$ is owned by a single cluster.  Can that owner change
next iteration?}
    \STATE $j \gets \mathrm{index\ of } \closest(\N_i)$
    \STATE $\ub(\N_i) \gets \ub(\N_i) + m_j$
    \STATE $\lb(\N_i) \gets \max \left(\lb(\N_i) - \max_i m_i, \min_{k \ne j}
d(c_k, c_j) / 2 \right)$
    \IF{$\ub(\N_i) < \lb(\N_i)$}
      \STATE \COMMENT{The owner cannot change next iteration.}
      \STATE $\canchange(\N_i) \gets \mathtt{false}$
    \ELSE
      \STATE \COMMENT{Tighten the upper bound and try to prune again.}
      \STATE $\ub(\N_i) \gets \min \left(\ub(\N_i), d_{\max}(\N_i, c_j)\right)$
      \STATE {\bf if} $\ub(\N_i) < \lb(\N_i)$ {\bf then} $\canchange(\N_i) \gets
\mathtt{false}$
    \ENDIF
  \ELSE
    \STATE $j \gets \mathrm{index\ of } \closest(\N_i)$
    \STATE $\ub(\N_i) \gets \ub(\N_i) + m_j$
    \STATE $\lb(\N_i) \gets \lb(\N_i) - \max_k m_k$
  \ENDIF
  \STATE \COMMENT{Recurse into each child.}
  \STATE{{\bf for each} child $\N_c$ of $\N_i${\bf , call}
\texttt{UpdateTree($\N_c$)}}
  \STATE \COMMENT{Try to determine points whose owner cannot change if $\N_i$
can change owners.}
  \IF{$\canchange(\N_i) = \mathtt{true}$}
    \FOR{$p_i \in \mathscr{P}_i$}
      \STATE $j \gets \mathrm{index\ of } \closest(p_i)$
      \STATE $\ub(p_i) \gets \ub(p_i) + m_j$
      \STATE $\lb(p_i) \gets \min\left( \lb(p_i) - \max_k m_k, \min_{k \ne j}
d(c_k, c_j) / 2 \right)$
      \IF{$\ub(p_i) < \lb(p_i)$}
        \STATE $\canchange(p_i) \gets \mathtt{false}$
      \ELSE
        \STATE \COMMENT{Tighten the upper bound and try again.}
        \STATE $\ub(p_i) \gets \min\left( \ub(p_i), d(p_i, c_j) \right)$
        \IF{$\ub(p_i) < \lb(p_i)$} 
          \STATE $\canchange(p_i) \gets \mathtt{false}$
        \ELSE
          \STATE \COMMENT{Point cannot be pruned.}
          \STATE $\ub(p_i) \gets \infty$
          \STATE $\lb(p_i) \gets \infty$
        \ENDIF
      \ENDIF
    \ENDFOR
  \ELSE
    \FOR{$p_i \in \mathscr{P}_i$ where $\canchange(p_i) = \mathtt{false}$}
      \STATE \COMMENT{Maintain upper and lower bounds for points whose owner
cannot change.}
      \STATE $j \gets \mathrm{index\ of } \closest(p_i)$
      \STATE $\ub(p_i) \gets \ub(p_i) + m_j$
      \STATE $\lb(p_i) \gets \lb(p_i) - \max_k m_k$
    \ENDFOR
  \ENDIF
  \IF{$\canchange(\cdot) = \mathtt{false}$ for all children $\N_c$ of
$\N_i$ and all points $p_i \in \mathscr{P}_i$}
    \STATE $\canchange(\N_i) \gets \mathtt{false}$
  \ENDIF
  \IF{$\canchange(\N_i) = \mathtt{true}$}
    \STATE $\pruned(\N_i) \gets 0$
  \ENDIF
\end{algorithmic}
\caption{\texttt{UpdateTree()} for dual-tree $k$-means.}
\label{alg:update_tree}
\end{algorithm*}

The first if statement (lines 4--10) catches the case where the parent cannot
change owner next iteration; in this case, the parent's upper bound and lower
bound can be taken as valid bounds.  In addition, the upper and lower bounds are
adjusted to account for cluster movement between iterations, so that the bounds
are valid for next iteration.

If the node $\mathscr{N}_i$ has an owner, the algorithm then attempts to use the
pruning rules established in Equations 4 and 6 in the main paper,
to determine if the owner of $\mathscr{N}_i$ can change next iteration.  If not,
$\canchange(\N_i)$ is set to \texttt{false} (line 18).  On the other hand, if
the pruning check fails, the upper bound is tightened and the pruning check is
performed a second time.  It is worth noting that $d_{\max}(\mathscr{N}_i, c_j)$
may not actually be less than the current value of $\ub(\N_i)$, which is why the
$\min$ is necessary.

After recursing into the children of $\N_i$, if $\N_i$ could have an owner
change, each point is individually checked using the same approach (lines
31--45).  However, there is a slight difference: if a point's owner can change,
the upper and lower bounds must be set to $\infty$ (lines 44--45).  This is only
necessary with points; \texttt{BaseCase()} does not take bounding information
from previous iterations into account, because no work can be avoided in that
way.

Then, we may set $\canchange(\N_i)$ to \texttt{false} if every point in $\N_i$
and every child of $\N_i$ cannot change owners (and the points and nodes do not
necessarily have to have the same owner).  Otherwise, we must set
$\pruned(\N_i)$ to $0$ for the next iteration.

\subsection{Coalescing the tree}

\begin{algorithm}[t!]
\begin{algorithmic}[1]
  \STATE {\bf Input:} tree $\mathscr{T}$
  \STATE {\bf Output:} coalesced tree $\mathscr{T}$
  \medskip
  \STATE \COMMENT{A depth-first recursion to hide nodes where
$\canchange(\cdot)$ is \texttt{false}.}
  \STATE $s \gets \{ \operatorname{root}(\mathscr{T}) \}$
  \WHILE{$|s| > 0$}
    \STATE $\N_i \gets s\mathtt{.pop()}$
    \medskip
    \STATE \COMMENT{Special handling is required for leaf nodes and the root
node.}
    \IF{$|\mathscr{C}_i| = 0$}
      \STATE {\bf continue}
    \ELSIF{$\N_i$ is the root node}
      \FOR{$\N_c \in \mathscr{C}_i$}
        \STATE $s\mathtt{.push(}\N_c\mathtt{)}$
      \ENDFOR
    \ENDIF
    \medskip
    \STATE \COMMENT{See if children can be removed.}
    \FOR{$\N_c \in \mathscr{C}_i$}
      \IF{$\canchange(\N_c) = \mathtt{false}$}
        \STATE remove child $\N_c$
      \ELSE
        \STATE $s\mathtt{.push(}\N_c\mathtt{)}$
      \ENDIF
    \ENDFOR
    \medskip
    \STATE \COMMENT{If only one child is left, then this node is unnecessary.}
    \IF{$|\mathscr{C}_i| = 1$}
      \STATE add child to $\parent(\N_i)$
      \STATE remove $\N_i$ from $\parent(\N_i)$'s children
    \ENDIF
  \ENDWHILE
  \medskip
  \STATE {\bf return} $\mathscr{T}$
\end{algorithmic}
\caption{\texttt{CoalesceTree()} for dual-tree $k$-means.}
\label{alg:coalesce}
\end{algorithm}

After \texttt{UpdateTree()} is called, the tree must be coalesced to remove any
nodes where $\canchange(\cdot) = \mathtt{false}$.  This can be accomplished via
a single pass over the tree.  A simple implementation is given in Algorithm
\ref{alg:coalesce}.  \texttt{DecoalesceTree()} may be implemented by
simply restoring a pristine copy of the tree which was cached right before
\texttt{CoalesceTree()} is called.

\subsection{Correctness proof}

As mentioned in the main document, a correctness proof is possible but
difficult.  We will individually prove the correctness of various pieces of the
dual-tree $k$-means algorithm, and then we will prove the main correctness
result.  For precision, we must introduce the exact definition of a space tree
and a pruning dual-tree traversal, as given by Curtin et~al.
\cite{curtin2013tree}.

\begin{defn}
A \textbf{space tree} on a dataset $S \in \Re^{N \times D}$ is an undirected,
connected, acyclic, rooted simple graph with the following properties:

\vspace*{-0.5em}
\begin{itemize}
  \item Each \textit{node} (or vertex), holds a number of points (possibly
zero) and is connected to one parent node and a number of child nodes (possibly
zero).
  \item There is one node in every space tree with no parent; this
is the \textit{root node} of the tree.
  \item Each point in $S$ is contained in at least one node.
  \item Each node $\mathscr{N}$ has a convex subset of
$\Re^{D}$ containing each point in that node and also the convex
subsets represented by each child of the node.
\end{itemize}
\end{defn}

\begin{defn}
\label{def:dtpt}
A {\it pruning dual-tree traversal} is a process that, given two space trees
$\mathscr{T}_q$ (the query tree, built on the query set $S_q$) and
$\mathscr{T}_r$ (the reference tree, built on the reference set $S_r$), will
visit combinations of nodes $(\mathscr{N}_q, \mathscr{N}_r)$ such that
$\mathscr{N}_q \in \mathscr{T}_q$ and $\mathscr{N}_r \in \mathscr{T}_r$ no more
than once, and call a function \texttt{Score($\mathscr{N}_q$, $\mathscr{N}_r$)}
to assign a score to that node.  If the score is $\infty$ (or above some bound),
the combination is pruned and no combinations ($\mathscr{N}_{qc}$,
$\mathscr{N}_{rc}$) such that $\mathscr{N}_{qc} \in \mathscr{D}^n_q$ and
$\mathscr{N}_{rc} \in \mathscr{D}^n_r$ are visited.  Otherwise, for every
combination of points ($p_q$, $p_r$) such that $p_q \in \mathscr{P}_q$ and $p_r
\in \mathscr{P}_r$, a function \texttt{BaseCase($p_q$, $p_r$)} is called.  If no
node combinations are pruned during the traversal, \texttt{BaseCase($p_q$,
$p_r$)} is called at least once on each combination of $p_q \in S_q$ and $p_r
\in S_r$.
\end{defn}

For more description and clarity on these definitions, refer to
\cite{curtin2013tree}.

\begin{lemma}
A pruning dual-tree traversal which uses \texttt{BaseCase()} as given in
Algorithm 2 in the main paper and \texttt{Score()} as given in Algorithm
3 in the main paper which starts with valid $\ub(\cdot)$, $\lb(\cdot)$,
$\pruned(\cdot)$, and $\closest(\cdot)$ for each node
$\N_i \in \mathscr{T}$, and $\ub(p_q) = \lb(p_q) = \infty$ for each point $p_q
\in S$, will satisfy the following conditions upon completion:

\vspace*{-0.4em}
\begin{itemize}
  \item For every $p_q \in S$ that is a descendant of a node $\N_i$ that has
been pruned ($\pruned(\N_i) = k - 1$), $\ub(\N_i)$ is an upper bound on the
distance between $p_q$ and its closest centroid, and $\closest(\N_i)$ is the
owner of $p_q$.

  \item For every $p_q \in S$ that is not a descendant of any node that has been
pruned, $\ub(p_q)$ is an upper bound on the distance between $p_q$ and its
closest centroid, and $\closest(p_q)$ is the owner of $p_q$.

  \item For every $p_q \in S$ that is a descendant of a node $\N_i$ that has
been pruned ($\pruned(\N_i) = k - 1$), $\lb(\N_i)$ is a lower bound on the
distance between $p_q$ and its second closest centroid.

  \item For every $p_q \in S$ that is not a descendant of any node that has been
pruned, $\min(\lb(p_q), \lb(\N_q))$ where $\N_q$ is a node such that $p_q \in
\mathscr{P}_q$ is a lower bound on the distance between $p_q$ and its second
closest centroid.
\end{itemize}
\label{lem:dt_correct}
\vspace*{-0.4em}
\end{lemma}

\begin{proof}
It is easiest to consider each condition individually.  Thus, we will first
consider the upper bound on the distance to the closest cluster centroid.
Consider some $p_q$ and suppose that the closest cluster centroid to $p_q$ is
$c^*$.

Now, suppose first that the point $p_q$ is a descendant point of a node $\N_q$
that has been pruned.  We must show, then, that $c^*$ is $\closest(\N_q)$.  Take
$R = \{ \N_{r0}, \N_{r1}, \ldots, \N_{rj} \}$ to be the set of reference nodes
visited during the traversal with $\N_q$ as a query node; that is, the
combinations $(\N_q, \N_{ri})$ were visited for all $\N_{ri} \in R$.  Any
$\N_{ri}$ is pruned only if

\vspace*{-0.8em}
\begin{equation}
d_{\min}(\N_q, \N_{ri}) > \ub(\N_i)
\end{equation}
\vspace*{-0.8em}

\noindent according to line 10 of \texttt{Score()}.  Thus, as
long as $\ub(\N_i)$ is a valid upper bound on the closest cluster distance for
every descendant point in $\N_q$, then no nodes are incorrectly pruned.  It is
easy to see that the upper bound is valid: initially, it is valid by assumption;
each time the bound is updated with some node $\N_{ri}$ (on lines
19 and 20), it is set to
$d_{\max}(\N_i, c)$ where $c$ is some descendant centroid of $\N_{ri}$.  This is
clearly a valid upper bound, since $c$ cannot be any closer to any descendant
point of $\N_i$ than $c^*$.  We may thus conclude that no node is incorrectly
pruned from $R$; we may apply this reasoning recursively to the $\N_q$'s
ancestors to see that no reference node is incorrectly pruned.

When a node is pruned from $R$, the number of pruned clusters for $\N_q$ is
updated: the count of all clusters not previously pruned by $\N_q$ (or its
ancestors) is added.  We cannot double-count the pruning of a cluster; thus the
only way that $\pruned(\N_q)$ can be equal to $k - 1$ is if every centroid
except one is pruned.  The centroid which is not pruned will be the nearest
centroid $c^*$, regardless of if $\closest(\N_q)$ was set during this traversal
or still holds its initial value, and therefore it must be true that $\ub(\N_q)$
is an upper bound on the distance between $p_q$ and $c^*$, and $\closest(\N_q) =
c^*$.

This allows us to finally conclude that if $p_q$ is a descendant of a node
$\N_q$ that has been pruned, then $\ub(\N_q)$ contains a valid upper bound on
the distance between $p_q$ and its closest cluster centroid, and
$\closest(\N_q)$ is that closest cluster centroid.

Now, consider the other case, where $p_q$ is not a descendant of any node that
has been pruned.  Take $\N_i$ to be any node containing $p_q$\footnote{Note that
the meaning here is not that $p_q$ is a descendant of $\N_i$ ($p_i \in
\mathscr{D}^p_i$), but instead that $p_q$ is held directly in $\N_i$: $p_q \in
\mathscr{P}_i$.}.  We have already reasoned that any cluster centroid node that
could possibly contain the closest cluster centroid to $p_q$ cannot have been
pruned; therefore, by the definition of pruning dual-tree traversal, we are
guaranteed that \texttt{BaseCase()} will be called with $p_q$ as the query point
and the closest cluster centroid as the reference point.  This will then cause
$\ub(p_q)$ to hold the distance to the closest cluster centroid---assuming
$\ub(p_q)$ is always valid, which it is even at the beginning of the traversal
because it is initialized to $\infty$---and $\closest(p_q)$ to hold the closest
cluster centroid.

Therefore, the first two conditions are proven.  The third and fourth
conditions, for the lower bounds, require a slightly different strategy.

There are two ways $\lb(\N_q)$ is modified: first, at line
14, when a node combination is pruned, and second, at
line 6 when the lower bound is taken from the parent.
Again, consider the set $R = \{ \N_{r0}, \N_{r1}, \ldots, \N_{rj} \}$ which is
the set of reference nodes visited during the traversal with $\N_q$ as a query
node.  Call the set of reference nodes that were pruned $R^p$.  At the end of
the traversal, then,

\vspace*{-1.3em}
\begin{eqnarray}
\lb(\N_q) &\le& \min_{\N_{ri} \in R^p} d_{\min}(\N_q, \N_{ri}) \\
  &\le& \min_{c_k \in C^p} d_{\min}(\N_q, c_k)
\end{eqnarray}
\vspace*{-1.3em}

\noindent where $C^p$ is the set of centroids that are descendants of nodes in
$R^p$.  Applying this reasoning recursively to the ancestors of $\N_q$ shows
that at the end of the dual-tree traversal, $\lb(\N_q)$ will contain a lower
bound on the distance between any descendant point of $\N_q$ and any pruned
centroid.  Thus, if $\pruned(\N_q) = k - 1$, then $\lb(\N_q)$ will contain a
lower bound on the distance between any descendant point in $\N_q$ and its
second closest centroid.  So if we consider some point $p_q$ which is a
descendant of $\N_q$ and $\N_q$ is pruned ($\pruned(\N_q) = k - 1$), then
$\lb(\N_q)$ is indeed a lower bound on the distance between $p_q$ and its second
closest centroid.

Now, consider the case where $p_q$ is not a descendant of any node that has been
pruned, and take $\N_q$ to be some node that owns $p_q$ (that is, $p_q \in
\mathscr{P}_q$).  In this case, \texttt{BaseCase()} will be called with every
centroid that has not been pruned.  So $\lb(\N_q)$ is a lower bound on the
distance between $p_q$ and every pruned centroid, and $\lb(p_q)$ will be a lower
bound on the distance between $p_q$ and the second-closest non-pruned centroid,
due to the structure of the \texttt{BaseCase()} function.  Therefore,
$\min(\lb(p_q), \lb(\N_q))$ must be a lower bound on the distance between $p_q$
and its second closest centroid.

Finally, we may conclude that each item in the theorem holds.
\end{proof}

Next, we must prove that \texttt{UpdateTree()} functions correctly.

\begin{lemma}
In the context of Algorithm 1 in the main paper, given a
tree $\mathscr{T}$ with all associated bounds $\ub(\cdot)$ and $\lb(\cdot)$ and
information $\pruned(\cdot)$, $\closest(\cdot)$, and $\canchange(\cdot)$, a run
of \texttt{UpdateTree()} as given in Algorithm \ref{alg:update_tree} will have
the following effects:

\begin{itemize}
  \item For every node $\N_i$, $\ub(\N_i)$ will be a valid upper bound on the
distance between any descendant point of $\N_i$ and its nearest centroid next
iteration.

  \item For every node $\N_i$, $\lb(\N_i)$ will be a valid lower bound on the
distance between any descendant point of $\N_i$ and any pruned centroid next
iteration.

  \item A node $\N_i$ will only have $\canchange(\N_i) = \mathtt{false}$ if the
owner of any descendant point of $\N_i$ cannot change next iteration.

  \item A point $p_i$ will only have $\canchange(p_i) = \mathtt{false}$ if the
owner of $p_i$ cannot change next iteration.

  \item Any point $p_i$ with $\canchange(p_i) = \mathtt{true}$ that does not
belong to any node $\N_i$ with $\canchange(\N_i) = \mathtt{false}$ will have
$\ub(p_i) = \lb(p_i) = \infty$, as required by the dual-tree traversal.

  \item Any node $\N_i$ with $\canchange(\N_i) = \mathtt{false}$ at the end of
\texttt{UpdateTree()} will have $\pruned(\N_i) = 0$.
\end{itemize}
\vspace*{-0.6em}
\label{lem:update_correct}
\end{lemma}

\begin{proof}
Each point is best considered individually.  It is important to remember during
this proof that the centroids have been updated, but the bounds have not.  So
any cluster centroid $c_i$ is already set for next iteration.  Take $c^l_i$ to
mean the cluster centroid $c_i$ {\it before} adjustment (that is, the old
centroid).  Also take $\ub^l(\cdot)$, $\lb^l(\cdot)$, $\pruned^l(\cdot)$, and
$\canchange^l(\cdot)$ to be the values at the time \texttt{UpdateTree()} is
called, before any of those values are changed.  Due to the assumptions in the
statement of the lemma, each of these quantities is valid.

Suppose that for some node $\N_i$, $\closest(\N_i)$ is some cluster $c_j$.  For
$\ub(\N_i)$ to be valid for next iteration, we must guarantee that $\ub(\N_i)
\ge \max_{p_q \in \mathscr{D}^p_q} d(p_q, c_j)$ at the end of
\texttt{UpdateTree()}.  There are four ways $\ub(\N_i)$ is updated: it may be
taken from the parent and adjusted (line 8), it may be adjusted before a prune
attempt (line 14), it may be tightened after a failed prune attempt (line 21),
or it may be adjusted without a prune attempt (line 25).  If we can show that
each of these four ways always results in $\ub(\N_i)$ being valid, then the
first condition of the theorem holds.

If $\ub(\N_i)$ is adjusted in line 14 or 25, the resulting value of $\ub(\N_i)$,
assuming $\closest(\N_i) = c_j$, is

\vspace*{-1em}
\begin{eqnarray}
\ub(\N_i) &=& \ub^l(\N_i) + m_j \\
  &\ge& \max_{p_q \in \mathscr{D}^p_q} d(p_q, c^l_j) + m_j \\
  &\ge& \max_{p_q \in \mathscr{D}^p_q} d(p_q, c_j)
\end{eqnarray}
\vspace*{-1em}

\noindent where the last step follows by the triangle inequality: $d(c_j, c^l_j)
= m_j$.  Therefore those two updates to $\ub(\N_i)$ result in valid upper bounds
for next iteration.  If $\ub(\N_i)$ is recalculated, in line
21, then we are guaranteed that $\ub(\N_i)$ is valid because

\vspace*{-1em}
\begin{equation}
d_{\max}(\N_i, c_j) \ge \max_{p_q \in \mathscr{D}^p_q} d(p_q, c_j).
\end{equation}
\vspace*{-1em}

We may therefore conclude that $\ub(\N_i)$ is correct for the root of the tree,
because line 8 can never be reached.  Reasoning
recursively, we can see that any upper bound passed from the parent must be
valid.  Therefore, the first item of the lemma holds.

Next, we will consider the lower bound, using a similar strategy.  We must show
that

\vspace*{-0.9em}
\begin{equation}
\lb(\N_i) \le \min_{p_q \in \mathscr{D}^p_q} \min_{c_p \in C_p} d(p_q, c_p)
\end{equation}
\vspace*{-0.9em}

\noindent where $C_p$ is the set of centroids pruned by $\N_i$ and ancestors
during the last dual-tree traversal.  The lower bound can be taken from the
parent in line 9 and adjusted, it can be adjusted before a prune attempt in line
15 or in a similar way without a prune attempt in line 26.  The last adjustment
can easily be shown to be valid:

\begin{eqnarray}
\lb(\N_i) &=& \lb^l(\N_i) - \max_k m_k \\
  &\le& \left( \min_{p_q \in \mathscr{D}^p_q} \min_{c_p \in C_p} d(p_q, c^l_p)
\right) - \max_k m_k \\
  &\le& \min_{p_q \in \mathscr{D}^p_q} \min_{c_p \in C_p} d(p_q, c_p)
\end{eqnarray}
\vspace*{-0.8em}

\noindent which follows by the triangle inequality: $d(c^l_p, c_p) \le \max_k
m_k$.  Line 15 is slightly more complex; we must also consider the term $\min_{k
\ne j} d(c_k, c_j) / 2$.  Suppose that

\vspace*{-0.8em}
\begin{equation}
\min_{k \ne j} d(c_k, c_j) / 2 > \lb^l(\N_i) + \max_k m_k.
\end{equation}
\vspace*{-0.8em}

We may use the triangle inequality ($d(p_q, c_k) \le d(c_j, c_k) + d(p_q, c_j)$)
to show that if this is true, the second closest centroid $c_k$ is such that
$d(p_q, c_k) > 2 d(c_k, c_j)$ and therefore $\min_{k \ne j} d(c_k, c_j) / 2$ is
also a valid lower bound.  We can lastly use the same recursive argument from
the upper bound case to show that the second item of the lemma holds.

Showing the correctness of $\canchange(\N_i)$ is straightforward: we know that
$\ub(\N_i)$ and $\lb(\N_i)$ are valid for next iteration by the time
any checks to set $\canchange(\N_i)$ to \texttt{false} happens, due to the
discussion above.  The situations where $\canchange(\N_i)$ is set to
\texttt{false}, in line 18 and 22, are simply applications of Equations
4 and 6 in the main paper, and are therefore valid.  There are
two other ways $\canchange(\N_i)$ can be set to \texttt{false}.  The first is on
line 10, and this is easily shown to be valid: if
a parent's owner cannot change, then a child's owner cannot change either.  The
other way to set $\canchange(\N_i)$ to \texttt{false} is in line 53.  This is only possible if all
points in $\mathscr{P}_i$ and all children of $\N_i$ have $\canchange(\cdot)$
set to \texttt{false}; thus, no descendant point of $\N_i$ can change owner next
iteration, and we may set $\canchange(\N_i)$ to \texttt{false}.

Next, we must show that $\canchange(p_i) = \mathtt{false}$ only if the owner of
$p_i$ cannot change next iteration.  If $\canchange^l(p_i) = \mathtt{true}$,
then due to Lemma \ref{lem:dt_correct}, $\ub^l(p_i)$ and $\lb^l(p_i)$ will be valid
bounds.  In this case, we may use similar reasoning to show that $\ub(p_i)$ and
$\lb(p_i)$ are valid, and then we may see that the pruning attempts at line
35 and 40 are valid.  Now, consider the other
case, where $\canchange^l(p_i) = \mathtt{false}$.  Then, $\ub^l(p_i)$ and
$\lb^l(p_i)$ will not have been modified by the dual-tree traversal, and will
hold the values set in the previous run of \texttt{UpdateTree()}.  As long as
those values are valid, then the fourth item holds.

The checks to see if $\canchange(p_i)$ can be set to \texttt{false} (from lines
31 to 45) are only reached if $\canchange(\N_i)$ is
\texttt{true}.  We already have shown that $\ub(p_i)$ and $\lb(p_i)$ are set
correctly in that stanza.  The other case is if $\canchange(\N_i)$ is
\texttt{false}.  In this case, lines 47 to 51 are reached.  It is easy to see
using similar reasoning to all previous cases that these lines result in valid
$\ub(p_i)$ and $\lb(p_i)$.  Therefore, the fourth item does hold.

The fifth item is taken care of in line 44 and 45.  Given some point $p_i$ with
$\canchange(p_i) = \mathtt{true}$, and where $p_i$ does not belong to any node
$\N_i$ where $\canchange(\N_i) = \mathtt{false}$, these two lines must be
reached, and therefore the fifth item holds.

The last item holds trivially---any node $\N_i$ where $\canchange(\N_i) =
\mathtt{true}$ has $\pruned(\N_i)$ set to $0$ on line 55.
\end{proof}

Showing that \texttt{CoalesceTree()}, \texttt{DecoalesceTree()}, and
\texttt{UpdateCentroids()} function correctly follows directly from the
algorithm descriptions.  Therefore, we are ready to show the main correctness
result.

\begin{thm}
A single iteration of dual-tree $k$-means as given in Algorithm
1 in the main paper will produce exactly the same results as the standard
brute-force $O(kN)$ implementation.
\end{thm}

\begin{proof}
We may use the previous lemmas to flesh out our earlier proof sketch.

First, we know that the dual-tree algorithm (line
9) produces correct results for $\ub(\cdot)$,
$\lb(\cdot)$, $\pruned(\cdot)$, and $\closest(\cdot)$ for every point and node,
due to Lemma \ref{lem:dt_correct}.
Next, we know that \texttt{UpdateTree()} maintains the correctness of those four
quantities and only marks $\canchange(\cdot)$ to \texttt{false} when the node or
point truly cannot change owner, due to Lemma \ref{lem:update_correct}.  Next,
we know from earlier discussion that \texttt{CoalesceTree()} and
\texttt{DecoalesceTree()} do not affect the results of the dual-tree algorithm
because the only nodes and points removed are those where $\canchange(\cdot) =
\mathtt{false}$.  We also know that \texttt{UpdateCentroids()} produces
centroids correctly.  Therefore, the results from Algorithm 1 in the main paper
are identical to those of a brute-force $O(kN)$ $k$-means implementation.
\end{proof}

\subsection{Runtime bound proof}

We can use adaptive algorithm analysis techniques in order to bound the running
time of Algorithm 1 in the main paper, based on \cite{curtin2015plug} and
\cite{langford2006}.  This analysis depends on the {\it expansion constant},
which is a measure of intrinsic dimension defined below, originally from
\cite{karger2002finding}.

\begin{defn}
\label{def:int_dim}
Let $B_S(p, \Delta)$ be the set of points in $S$ within a closed ball of radius
$\Delta$ around some $p \in S$ with respect to a metric $d$:

\vspace*{-1.2em}
\begin{equation}
B_S(p, \Delta) = \{ r \in S \colon d(p, r) \leq \Delta \}.
\end{equation}
\vspace*{-1.2em}

Then, the {\bf expansion constant} of $S$ with respect to the metric $d$ is the
smallest $c \ge 2$ such that

\vspace*{-1.2em}
\begin{equation}| B_S(p, 2 \Delta) | \le c | B_S(p, \Delta) |\ \forall\ p \in S,\
\forall\ \Delta > 0.
\end{equation}
\vspace*{-1.2em}
\end{defn}

The expansion constant is a bound on the number of points which fall into balls
of increasing sizes.  A low expansion constant generally means that search tasks
like nearest neighbor search can be performed quickly with trees, whereas a high
expansion constant implies a difficult dataset.  Thus, if we assume a bounded
expansion constant like in previous theoretical works \cite{langford2006,
ram2009, karger2002finding, curtin2014dual, curtin2015plug}, we may assemble a
runtime bound that reflects the difficulty of the dataset.

Our theoretical analysis will concern the cover tree in particular.  The cover
tree is a complex data structure with
appealing theoretical properties.  We will only summarize the relevant
properties here.  Interested readers should consult the original cover tree
paper \cite{langford2006} and later analyses \cite{ram2009, curtin2015plug} for
a complete understanding.

A cover tree is a leveled tree; that is, each cover tree node $\mathscr{N}_i$ is
associated with an integer scale $s_i$.  The node with largest scale is the root
of the tree; each node's scale is greater than its children's.  Each node
$\mathscr{N}_i$ holds one point $p_i$, and every descendant point of
$\mathscr{N}_i$ is contained in the ball centered at $p_i$ with radius $2^{s_r +
1}$.  Further, every cover tree satisfies the following three invariants
\cite{langford2006}:

\begin{itemize}
\item {\it (Nesting.)}  When a point $p_i$ is held in a node at some scale
$s_i$, then each smaller scale will also have a node containing $p_i$.

\item {\it (Covering tree.)}  For every point $p_i$ held in a node
$\mathscr{N}_i$ at scale $s_i$, there exists a node with point $p_j$ and scale
$s_i + 1$ which is the parent of $\mathscr{N}_i$, and $d(p_i, p_j) < 2^{s_i +
1}$.

\item {\it (Separation.)}  Given distinct nodes $\mathscr{N}_i$ holding $p_i$
and $\mathscr{N}_j$ holding $p_j$ both at scale $s_i$, $d(p_i, p_j) > 2^{s_i}$.
\end{itemize}

A useful result shows there are $O(N)$ points in a cover tree (Theorem 1,
\cite{langford2006}).  Another measure of importance of a cover tree is the {\it
cover tree imbalance}, which aims to capture how well the data is distributed
throughout the cover tree.  For instance, consider a tree where the root, with
scale $s_r$, has two
nodes; one node corresponds to a single point and has scale $-\infty$, and the
other node has scale $s_r - 1$ and contains every other point in the dataset as
a descendant.  This is very imbalanced, and a tree with many situations like
this will not perform well for search tasks.  Below, we reiterate the definition
of cover tree imbalance from \cite{curtin2015plug}.

\begin{defn}
The {\it cover node imbalance} $i_n(\mathscr{N}_i)$ for a cover tree node
$\mathscr{N}_i$ with scale $s_i$ in the cover tree $\mathscr{T}$ is defined as
the cumulative number of missing levels between the node and its parent
$\mathscr{N}_p$ (which has scale $s_p$).  If
the node is a leaf child (that is, $s_i = -\infty$), then number of missing
levels is defined as the difference between $s_p$ and $s_{\min} - 1$ where
$s_{\min}$ is the smallest scale of a non-leaf node in $\mathscr{T}$.  If
$\mathscr{N}_i$ is the root of the tree, then the cover node imbalance is 0.
Explicitly written, this calculation is

\begin{equation}
i_n(\mathscr{N}_i) = \begin{dcases*}
  s_p - s_i - 1 & if $\mathscr{N}_i$ is not a \\
  & leaf and not \\
  & the root node \\
  \max(s_p - s_{\min} - 1, \; 0) & if $\mathscr{N}_i$ is a leaf \\
  0 & if $\mathscr{N}_i$ is the root. \\
  \end{dcases*}
  \label{eqn_node_imbalance}
\end{equation}
\end{defn}

This simple definition of cover node imbalance is easy to calculate, and using
it, we can generalize to a measure of imbalance for the full tree.

\begin{defn}
\label{def:imbalance}
The {\it cover tree imbalance} $i_t(\mathscr{T})$ for a cover tree $\mathscr{T}$
is defined as the cumulative number of missing levels in the tree.  This can be
expressed as a function of cover node imbalances easily:

\begin{equation}
i_t(\mathscr{T}) = \sum_{\mathscr{N}_i \in \mathscr{T}} i_n(\mathscr{N}_i).
\end{equation}
\end{defn}

Bounding $i_t(\mathscr{T})$ is non-trivial, but empirical results suggest that
imbalance scales linearly with the size of the dataset, when the expansion
constant is well-behaved.  A bound on $i_t(\mathscr{T})$ is still an open
problem at the time of this writing.

With these terms introduced, we may introduce a slightly adapted result from
\cite{curtin2015plug}, which bounds the running time of nearest neighbor search.

\begin{thm}
(Theorem 2, \cite{curtin2015plug}.)  Using cover trees, the standard cover tree
pruning dual-tree traversal, and the
nearest neighbor search \texttt{BaseCase()} and \texttt{Score()} as given in
Algorithms 2 and 3 of \cite{curtin2015plug}, respectively, and also
given a reference set $S_r$ with expansion constant $c_r$, and a query set
$S_q$, where the range of pairwise distances in $S_r$ is completely contained in
the range of pairwise distances in $S_q$, the running time of nearest neighbor
search is bounded by $O(c_r^4 c_{qr}^5 (N + i_t(\mathscr{T}_q)))$, where
$c_{qr} = \max((\max_{p_q \in S_q} c_r'), c_r)$, where $c_r'$ is the expansion
constant of the set $S_r \cup \{ p_q \}$.
\label{thm:nns}
\end{thm}

Now, we may adapt this result slightly.

\begin{thm}
The dual-tree $k$-means algorithm with \texttt{BaseCase()} as in Algorithm
2 in the main paper and \texttt{Score()} as in Algorithm 3 in the main paper, with a
point set $S_q$ that has expansion constant $c_q$ and size $N$, and $k$ centroids
$C$ with expansion constant $c_k$, takes no more than $O(c_k^4 c_{qk}^5 (N +
i_t(\mathscr{T}_q)))$ time.
\label{thm:dtkm}
\end{thm}

\begin{proof}
Both \texttt{Score()} and \texttt{BaseCase()} for dual-tree $k$-means can be
performed in $O(1)$ time.  In addition, the pruning of \texttt{Score()} for
dual-tree $k$-means is at least as tight as \texttt{Score()} for nearest
neighbor search: the pruning rule in Equation 2 in the main paper is equivalent to
the pruning rule for nearest neighbor search.  Therefore, dual-tree $k$-means
can visit no more nodes than nearest neighbor search would with query set $S_q$
and reference set $C$.  Lastly, note that the range of pairwise distances of $C$
will be entirely contained in the range of pairwise distances in $S_q$, to see
that we can use the result of Theorem \ref{thm:nns}.  Adapting that result,
then, yields the statement of the algorithm.
\end{proof}

The expansion constant of the centroids, $c_k$, may be understood as the
intrinsic dimensionality of the centroids $C$.  During each iteration, the
centroids change, so those iterations that have centroids with high intrinsic
dimensionality cannot be bounded as tightly.  More general measures of intrinsic
dimensionality, such as those recently proposed by Houle
\cite{houle2013dimensionality}, may make the connection between $c_q$ and $c_k$
clear.

Next, we turn to bounding the entire algorithm.

\begin{thm}
A single iteration of the dual-tree $k$-means algorithm on a dataset $S_q$ using
the cover tree $\mathscr{T}$, the standard cover tree pruning dual-tree
traversal, \texttt{BaseCase()} as given in Algorithm 2 in the main paper,
\texttt{Score()} as given in Algorithm 3 in the main paper, will take no more than

\vspace*{-0.3em}
\begin{equation}
O(c_k^4 c_{qk}^5 (N + i_t(\mathscr{T})) + c_k^9 k \log k)
\end{equation}

\noindent time, where $c_k$ is the expansion constant of the centroids, $c_{qk}$
is defined as in Theorem \ref{thm:dtkm}, and $i_t(\mathscr{T})$ is the imbalance of
the tree as defined in Definition \ref{def:imbalance}.
\end{thm}

\begin{proof}
Consider each of the steps of the algorithm individually:

\begin{itemize}
  \item \texttt{CoalesceNodes()} can be performed in a single pass of the cover
tree $\mathscr{N}$, which takes $O(N)$ time.

  \item Building a tree on the centroids ($\mathscr{T}_c$) takes $O(c_k^6 k \log
k)$ time due to the result for cover tree construction time \cite{langford2006}.

  \item The dual-tree algorithm takes $O(c_k^4 c_{qk}^5 (N + i_t(\mathscr{T})))$
time due to Theorem \ref{thm:dtkm}.

  \item \texttt{DecoalesceNodes()} can be performed in a single pass of the
cover tree $\mathscr{N}$, which takes $O(N)$ time.

  \item \texttt{UpdateCentroids()} can be performed in a single pass of the
cover tree $\mathscr{N}$, so it also takes $O(N)$ time.

  \item \texttt{UpdateTree()} depends on the calculation of how much each
centroid has moved; this costs $O(k)$ time.  In addition, we must find the
nearest centroid of every centroid; this is nearest neighbor search, and we may
use the runtime bound for monochromatic nearest neighbor search for cover trees
from \cite{ram2009}, so this costs $O(c_k^9 k)$ time.  Lastly, the actual tree
update visits each node once and iterates over each point in the node.  Cover
tree nodes only hold one point, so each visit costs $O(1)$ time, and with $O(N)$
nodes, the entire update process costs $O(N)$ time.  When we consider the
preprocessing cost too, the total cost of \texttt{UpdateTree()} per iteration is
$O(c_k^9 k + N)$.
\end{itemize}

We may combine these into a final result:

\vspace*{-1em}
\begin{eqnarray}
O(N) + O(c_k^6 k \log k) + O(c_k^4 c_{qk}^5 (N + i_t(\mathscr{T}))) + \nonumber \\
\ \ O(N) + 
O(N) + O(c_k^9 k + N)
\end{eqnarray}
\vspace*{-1em}

\noindent and after simplification, we get the statement of the theorem:

\vspace*{-1em}
\begin{equation}
O(c_k^4 c_{qk}^5 (N + i_t(\mathscr{T})) + c_k^9 k \log k).
\end{equation}
\end{proof}

Therefore, we see that under some assumptions on the data, we can bound the
runtime of the dual-tree $k$-means algorithm to something tighter than $O(kN)$
per iteration.  As expected, we are able to amortize the cost of $k$ across all
$N$ nodes, giving amortized $O(1)$ search for the nearest centroid per point in
the dataset.  This is similar to the results for nearest neighbor search, which
obtain amortized $O(1)$ search for a single query point.  Also similar to the
results for nearest neighbor search is that the search time may, in the worst
case, degenerate to $O(kN + k^2)$ when the assumptions on the dataset are not
satisfied.  However, empirical results \cite{ram2009rank, gray2001nbody,
march2010euclidean, langford2006} show that well-behaved datasets are common in
the real world, and thus degeneracy of the search time is uncommon.

Comparing this bound with the bounds for other algorithms is somewhat difficult;
first, none of the other algorithms have bounds which are adaptive to the
characteristics of the dataset.  It is possible that the blacklist algorithm
could be refactored to use the cover tree, but even if that was done it is not
completely clear how the running time could be bounded.  How to apply the
expansion constant to an analysis of Hamerly's algorithm and Elkan's algorithm
is also unclear at the time of this writing.

Lastly, the bound we have shown above is potentially loose.  We have reduced
dual-tree $k$-means to the problem of nearest neighbor search, but our pruning
rules are tighter.  Dual-tree nearest neighbor search assumes that every query
node will be visited (this is where the $O(N)$ in the bound comes from), but
dual-tree $k$-means can prune a query node entirely if all but one cluster is
pruned (Strategy 2).  These bounds do not take this pruning strategy into
account, and they also do not consider the fact that coalescing the tree can
greatly reduce its size.  These would be interesting directions for future
theoretical work.

\end{document}